\documentclass[preprint,number,sort&compress,12pt]{elsarticle}

\usepackage[utf8]{inputenc}
\usepackage[T1]{fontenc}
\usepackage{amsmath,amssymb}
\usepackage{amsthm}
\usepackage{booktabs}
\usepackage{graphicx}
\usepackage{microtype}
\usepackage{xcolor}
\usepackage[colorlinks=true,linkcolor=blue!60!black,citecolor=blue!60!black,urlcolor=blue!60!black]{hyperref}
\usepackage{enumitem}
\setlist{itemsep=2pt,topsep=4pt}

\theoremstyle{plain}
\newtheorem{proposition}{Proposition}
\newtheorem{claim}{Claim}

\newcommand{\graft}{\textsc{Graft}}
\newcommand{\fgraft}{\textsc{FGraft}}
\newcommand{\efh}{\ensuremath{\mathit{ef}_h}}
\newcommand{\dq}[1]{\ensuremath{d@#1}}
\newcommand{\G}{\,G}

\begin{document}

\begin{frontmatter}

\title{Batched Feedback and the Random-Access Wall\\in Search-Based Graph Construction}

\author{\'Edgar Ch\'avez}
\ead{elchavez@cicese.edu.mx}
\affiliation{organization={Computer Science Department, CICESE},
              city={Ensenada}, country={M\'exico}}

\begin{abstract}
Navigable graphs for nearest-neighbor search are built either
\emph{incrementally}, each inserted point searching a graph that mutates as
construction proceeds, or in \emph{batch} over a fixed substrate, which buys
determinism and parallelism at a price in build time. We measure that price
and find where it comes from. Instrumenting the build-side distance counters
of a tuned Vamana and of PiPNN, and building every system repeatedly in a
paired design on one 64-thread machine, we separate build time into
\emph{work} (distance evaluations) and \emph{cost per evaluation}. Letting the
batch builder's substrate mutate in $B$ synchronous blocks (\fgraft) recovers
the feedback loop of incremental construction while the graph stays a
function of (data, seed, parameters, $B$): eight trees with batched feedback
match thirty-two frozen trees, the build does $0.88\times$ Vamana's distance
work, and coarse batching ($B{=}8$) suffices. Yet the build takes
$1.55\times$ Vamana's wall-clock, because the mutating substrate costs more
per evaluation. Pushing further, we find a wall that no search-based builder
crosses: every one of them, incremental or batched, evaluates distances
\emph{adaptively}, one dependent random access at a time, and runs 20--27$\times$
below a dense kernel on the same machine at $d\approx100$, five to six times
of it with the data resident in cache. The gap is not a low-dimension
artifact: measured against dimension, an adaptive evaluation costs $d^{1.04}$
and a blocked one $d^{0.63}$, so the wall grows as $d^{0.4}$ and is twice as
high at $d{=}960$ as at $d{=}128$. PiPNN's order-of-magnitude build advantage is that
kernel: it evaluates as many distances per point as our beam does, but as
fixed-in-advance dense blocks. We show the beam cannot be batched after the
fact (the useful density of a lockstep block is 2--3\%), state the wall as a
two-ceiling roofline, and delimit its scope: it binds whenever the distance is
a black box or the evaluation order is data-dependent, and it is absent only
for decomposable distances evaluated on pair sets fixed in advance.
\end{abstract}


\end{frontmatter}

\section{Introduction}\label{sec:intro}

\graft~\cite{chavez2026graft} removes mutation from search-based graph
construction. It builds a forest of $T$ spatial-approximation
trees~\cite{navarro2002sat}, unions their edges into a fixed, redundant
\emph{scaffold}, and \emph{harvests}: every point, in parallel, beam-searches
the frozen scaffold with itself as the query, and the vertices that search
expands are occlusion-pruned into its adjacency list. The construction is
embarrassingly parallel and bitwise deterministic at any thread count, and at
matched recall its graphs are within a few percent of tuned
Vamana~\cite{subramanya2019diskann,manohar2024parlayann} in distances per
query. It is also slower to build. This paper asks where that time goes and
finds two answers at two different scales.

The first answer is algorithmic and partly good news. A frozen substrate must
be routable, at full scale, for the first point it serves, and that
routability is bought with redundancy: more trees. On GloVe the frozen design
needs $T{=}32$, and the trees alone cost $75\%$ of what Vamana spends in
total. An incremental builder never pays this: every insertion improves the
substrate the next one searches. We show that this feedback loop is separable
from the serial dependency chain that produces it. \fgraft{} mutates the
substrate in $B$ synchronous blocks; within a block it is frozen and the
writes are disjoint, so the graph remains a pure function of
(data, seed, parameters, $B$). Measured with a paired, repeated-build protocol
on one 64-thread machine (\S\ref{sec:results}), eight trees with batched
feedback reach the quality of thirty-two frozen trees, the build performs
$0.88\times$ Vamana's distance work, $B{=}8$ captures the whole gain, and the
build time drops to $0.58\times$ that of the frozen quality profile. It does
\emph{not} drop below Vamana's: $1.55\times$ [1.55, 1.56], because the
mutating, denser substrate costs $1.8\times$ per distance evaluation where the
frozen one costs $1.3\times$.

The second answer is architectural and is the contribution we consider
lasting. Chasing the remaining factor, we measured what a search-based
builder can do at all (\S\ref{sec:wall}). Every such builder, incremental
(HNSW, Vamana) or batched (\graft, \fgraft), evaluates distances
\emph{adaptively}: which vector is touched next depends on the last result.
Each evaluation is therefore a dependent random access, and its rate is
bounded twice, by the per-access bookkeeping that adaptivity forces and by
the random-access bandwidth of the memory that holds the data. On our machine
two independently engineered search builders sit within $1.3\times$ of each
other and 20--27$\times$ below a dense GEMM at $d\approx100$, and five to six
times of that gap is already present with the data in cache. Sweeping the
dimension from $96$ to $960$ shows the ratio growing, not shrinking, as
$d^{0.4}$ (Section~\ref{sec:dim}): the corpora with the largest gap are the
embedding corpora. PiPNN~\cite{rubel2026pipnn}, whose
build is $14\times$ faster than Vamana's, evaluates as many distances per
point as our beam does; its advantage is that its pair sets are fixed in
advance and computed as dense blocks. We then ask whether the beam can be
batched after the fact, and measure that it cannot: the useful density of a
lockstep block of neighboring points is 2--3\%, the beams of adjacent points
diverge immediately, and the SAT subtrees that define coherent blocks are not
compact in the scaffold.

The wall has a clean statement and a clean scope. \emph{Refining a graph
construction whose distance evaluations are not coherent, because the
distance is a black box or because the evaluation order is data-dependent,
is bounded by random-access throughput, not by arithmetic.} It binds for
every metric-space index and for every search-based vector index; it is
absent only when the distance decomposes into inner products \emph{and} the
pair set is known before it is evaluated. Counting distance evaluations,
the natural cost model of metric-space search, mismeasures the vector case
by a factor of twenty.

\section{Setting and protocol}\label{sec:setting}

\paragraph{Systems} \graft{} and \fgraft{} are the \texttt{fg} binary of
the public \texttt{graft-ann} repository (release v0.3 plus the mutation
patch). Baselines are ParlayANN's Vamana~\cite{manohar2024parlayann} with the
authors' parameters (GloVe $R{=}100$, $L{=}200$, $\alpha{=}1.0$, two passes;
SIFT $R{=}64$, $L{=}128$, $\alpha{=}1.15$, two passes), PiPNN with the same
$R$, $L$~\cite{rubel2026pipnn}, and hnswlib~\cite{malkov2020hnsw} at
$M\in\{16,32\}$, $ef_c{=}200$. Datasets are ann-benchmarks'
GloVe-100~\cite{pennington2014glove} ($1.18$M, angular, L2-normalized so that
every system solves the same Euclidean problem) and SIFT-128
($10^6$)~\cite{aumuller2020annbenchmarks};
all systems are graded against the same brute-force answer key.

\paragraph{Machine} A 4-socket Xeon E7-4809 v3 (Haswell, 2.0 GHz, AVX2,
32 cores / 64 threads, 20 MB L3 per socket, 1.5 TB RAM), otherwise idle,
with no thermal drift (build-time IQR $\le 0.5\%$ over ten repetitions).

\paragraph{Cost model} Build cost is reported two ways. \emph{Work} is the
number of distance evaluations, read from each system's own build-side
counter: \texttt{fg} counts every evaluation; ParlayANN's Vamana computes the
count and we patched one line to print it; for PiPNN we added counters to its
three evaluating phases (\S\ref{sec:pipnn}). \emph{Time} is each system's own
build timer (data already in RAM), measured with a paired, repeated-build
design: ten blocks, every system built once per block in a seeded random
order, ratios reported as Hodges--Lehmann estimates with exact 95\% confidence
intervals and exact Wilcoxon signed-rank tests over the paired blocks. The
quotient of the two ratios is the \emph{cost per evaluation} relative to
Vamana. Quality is distances per query at a matched recall target,
interpolated between adjacent beam widths (\dq{r}); a dash means the sweep did
not bracket the target. Local completeness is the fraction of a point's true
10-NN present in its adjacency list.

\section{Batched feedback}\label{sec:method}

Let $S_0$ be the \graft{} scaffold: the union of $T$ SAT trees' edges,
symmetrized. Partition the points into $B$ blocks by a fixed rule ($p \bmod
B$). For $b = 0,\dots,B-1$: harvest the points of block $b$ against $S_b$,
exactly as \graft{} does (beam search of width \efh{} from the tree roots,
occlusion prune~\cite{harwood2016fanng,subramanya2019diskann} of the expanded
set to at most $R$ neighbors), writing each
point's row into its own slot; then
\[
S_{b+1} \;=\; S_b \;\cup\; \{(p,u),(u,p) : u \in A[p],\ p \in \text{block } b\}.
\]
The final graph is the harvested rows, symmetrized and re-pruned, plus the
tree-0 spine. Two design points matter. \emph{Union, never replace}: the
substrate needs redundancy and the output needs sparsity; they are different
objects, and replacing the scaffold by the pruned graph collapses quality
(measured in~\cite{chavez2026graft}). \emph{Coarse blocks}: $B$ substrate
rebuilds, not $n$, which is what keeps the scheme synchronous and cheap.

\begin{proposition}[Schedule independence, blocked]\label{prop:det}
Under the determinism rules of~\cite{chavez2026graft} (keyed randomness,
disjoint writes with canonical merges, a total order on candidates) and a
block assignment that is a function of the point identifier alone, the
\fgraft{} adjacency is a function of $(S, d, \mathit{params}, \mathit{seed},
B)$: bitwise identical for every thread count and every scheduling of the
parallel loops.
\end{proposition}
\begin{proof}[Proof sketch]
Induction on blocks. $S_0$ is deterministic by the cited contract. Given
$S_b$, every point of block $b$ searches a substrate that does not change
during the block and writes only its own row; search and prune are
deterministic functions of $(S_b, p)$. $S_{b+1}$ is a canonical merge of
deterministic inputs. No step depends on thread identity, count or
interleaving.
\end{proof}
The gate is executable: \texttt{fg --check-determinism} rebuilds at one and
at 64 threads and compares graph hashes. It passes for $B\in\{0,2,8,32\}$,
and $B{\le}1$ reproduces the published \graft{} hash, so the frozen path is
untouched.

\section{Results: work, time, and the cost of an evaluation}\label{sec:results}

\begin{table}[t]
\centering\small
\caption{GloVe-100, ten paired blocks, 64 threads. Medians; HL = Hodges--Lehmann
ratio of build time against Vamana with exact 95\% CI (every row $p=0.002$,
10/10 blocks agree). Quality is distances per query at matched recall,
relative to Vamana's $9{,}204$ / $13{,}325$.}
\label{tab:glove}
\resizebox{\linewidth}{!}{%
\begin{tabular}{@{}lrrrrrrrr@{}}
\toprule
arm & build s & G dist & compl. & \dq{0.95} & \dq{0.97} & HL time [CI] & work & cost/eval \\
\midrule
\graft{} T32/\efh 600 frozen & 390.6 & 63.87 & 0.590 & 1.021 & 1.047 & 2.68 [2.67, 2.70] & 2.02 & 1.33 \\
\graft{} T16/\efh 400 frozen & 181.3 & 30.28 & 0.550 & 1.103 & 1.164 & 1.24 [1.24, 1.25] & 0.96 & 1.29 \\
\textbf{\fgraft{} T8/\efh 400 $B{=}8$} & \textbf{226.3} & 27.66 & 0.570 & 1.058 & 1.068 & \textbf{1.55 [1.55, 1.56]} & 0.88 & 1.77 \\
\fgraft{} T8/\efh 400 $B{=}32$ & 287.2 & 28.86 & 0.574 & 1.050 & 1.062 & 1.97 [1.96, 1.98] & 0.91 & 2.16 \\
\midrule
Vamana R100/L200 2-pass & 145.9 & 31.56 & --- & 1 & 1 & 1 & 1 & 1 \\
PiPNN R100/L200 & 10.4 & 12.56 & --- & 1.447 & 1.413 & 0.071 & 0.40 & 0.18 \\
hnswlib M16 / M32 & 41.2 / 65.9 & --- & --- & --- & --- & 0.28 / 0.45 & --- & --- \\
\bottomrule
\end{tabular}}
\end{table}

\begin{table}[t]
\centering\small
\caption{SIFT-128, same protocol. Every \graft-family graph is above recall
$0.97$ at the bottom of the beam ladder, so quality is matched at $0.98$ /
$0.99$ (Vamana: $1{,}619$ / $2{,}172$).}
\label{tab:sift}
\resizebox{\linewidth}{!}{%
\begin{tabular}{@{}lrrrrrrrr@{}}
\toprule
arm & build s & G dist & compl. & \dq{0.98} & \dq{0.99} & HL time [CI] & work & cost/eval \\
\midrule
\graft{} T4/\efh 400 frozen & 52.1 & 5.15 & 0.483 & 1.074 & 1.067 & 1.02 [0.99, 1.03]$^\dagger$ & 0.47 & 2.16 \\
\graft{} T16/\efh 400 frozen & 104.3 & 12.08 & 0.498 & 1.002 & 0.990 & 2.01 [2.00, 2.03] & 1.09 & 1.85 \\
\fgraft{} T4/\efh 400 $B{=}8$ & 95.5 & 8.38 & 0.496 & 1.011 & 1.004 & 1.85 [1.83, 1.88] & 0.76 & 2.44 \\
\textbf{\fgraft{} T2/\efh 400 $B{=}8$} & \textbf{89.6} & 8.03 & 0.492 & 1.008 & 1.011 & \textbf{1.73 [1.72, 1.75]} & 0.73 & 2.39 \\
\fgraft{} T4/\efh 400 $B{=}32$ & 122.7 & 8.72 & 0.498 & 1.007 & 0.997 & 2.37 [2.35, 2.39] & 0.79 & 3.01 \\
\midrule
Vamana R64/L128 2-pass & 51.7 & 11.06 & --- & 1 & 1 & 1 & 1 & 1 \\
PiPNN R64/L128 & 8.8 & 9.88 & --- & 1.100 & 1.042 & 0.17 & 0.89 & 0.19 \\
hnswlib M16 / M32 & 29.4 / 37.2 & --- & --- & --- & --- & 0.57 / 0.72 & --- & --- \\
\bottomrule
\end{tabular}}
\\[2pt]{\footnotesize $^\dagger$Wilcoxon $p=0.30$, 5/10 blocks: parity.}
\end{table}

Tables~\ref{tab:glove} and~\ref{tab:sift} carry four findings.

\paragraph{Feedback substitutes for redundancy} On GloVe, $T{=}8$ with
$B{=}8$ reaches completeness $0.570$ and \dq{0.97} within $2\%$ of the frozen
$T{=}32$ profile, at $0.43\times$ its distance work and $0.58\times$
[0.58, 0.58] its build time. On SIFT, \emph{two} trees with feedback reach
the quality of sixteen frozen (parity with Vamana within $1\%$ at both
recall targets), where four frozen trees are $7\%$ behind.

\paragraph{The work is now below Vamana's; the time is not} \fgraft{} does
$0.88\times$ (GloVe) and $0.73\times$ (SIFT) Vamana's distance evaluations and
takes $1.55\times$ and $1.73\times$ its time. The reason is the last column:
the mutating substrate costs $1.8$--$2.4\times$ per evaluation against the
frozen scaffold's $1.3$--$2.2\times$. Later blocks search a denser graph, the
union of the scaffold and every finished row, and its beam is a worse pointer
chase. The distance count predicted $\approx 1.4\times$; the measurement says
$1.55\times$, with a confidence interval that excludes $1.5$.

\paragraph{Coarse batching suffices, and finer batching costs} $B{=}8$ and
$B{=}32$ are within $1\%$ in quality on both datasets, and $B{=}32$ is
$27$--$28\%$ slower. What a point needs is a substrate improved by many
predecessors, not by its immediate predecessor; eight rebuilds are enough,
and each additional rebuild densifies what the remaining points must search.

\paragraph{The regime dependence is a result} Where the frozen scaffold is
already cheap (SIFT's local intrinsic dimensionality is far below
GloVe's~\cite{aumuller2021lid}; $T{=}4$ does $0.47\times$ Vamana's work at time parity),
mutation removes little cost; it buys quality. Where the scaffold is
expensive (GloVe: $T{=}32$), mutation removes work and the per-evaluation
cost takes half of the gain back. On the decision rule fixed before the
campaign, full paper at $\le 1.5\times$ Vamana with parity quality, the
batched design lands just outside on both counts. The next section explains
why the missing factor is not an engineering matter.

\section{The random-access wall}\label{sec:wall}

\subsection{PiPNN's advantage, decomposed}\label{sec:pipnn}

PiPNN~\cite{rubel2026pipnn} builds its graph without searching: a
leader-based clustering tree places every point in a few hundred leaves of
$\approx 340$ points, each leaf is brute-forced with a dense rank-update GEMM,
the leaf $k$-NN graphs are unioned, and an $\alpha$-prune fixes the degree.
Its build is $14\times$ faster than Vamana's on GloVe and $5.9\times$ on
SIFT. Instrumenting the three evaluating phases gives Table~\ref{tab:pipnn}.

\begin{table}[t]
\centering\small
\caption{PiPNN's build distance evaluations by phase (our counters), and the
resulting throughput against the search builders of
Tables~\ref{tab:glove}--\ref{tab:sift}.}
\label{tab:pipnn}
\resizebox{\linewidth}{!}{%
\begin{tabular}{@{}lrrrrrr@{}}
\toprule
 & clustering & leaf GEMM & prune & total G & build s & G eval/s \\
\midrule
PiPNN GloVe & 3.21 & 8.49 & 0.86 & 12.56 & 10.35 & \textbf{1.21} \\
PiPNN SIFT  & 2.20 & 7.29 & 0.39 & 9.88 & 8.82 & \textbf{1.12} \\
\midrule
Vamana GloVe / SIFT & & & & 31.56 / 11.06 & 145.9 / 51.7 & 0.216 / 0.214 \\
\graft{} T16 / T4 frozen & & & & 30.28 / 5.15 & 181.3 / 52.1 & 0.167 / 0.099 \\
\fgraft{} $B{=}8$ (T8 / T2) & & & & 27.66 / 8.03 & 226.3 / 89.6 & 0.122 / 0.090 \\
\bottomrule
\end{tabular}}
\end{table}

Two things follow. First, PiPNN's advantage is only partly less work: on
GloVe it evaluates $0.40\times$ Vamana's distances, on SIFT $0.89\times$, and
in both cases about $10^4$ per point, the same order as our beam ($11.5$k at
\efh 400). Second, the rest is throughput: $1.1$--$1.2$ G evaluations per
second, $5$--$6\times$ the search builders, with $93\%$ of its evaluations
inside GEMMs. Each point sits in $30$ leaves; PiPNN buys quality with
redundant brute force it can afford because the kernel is cheap, and still
trails Vamana by $41$--$45\%$ in distances per query on GloVe.

\subsection{Throughput versus working set}\label{sec:sweep}

If the wall were DRAM bandwidth alone, the search builders' throughput would
collapse when the vectors stop fitting in cache. Table~\ref{tab:sweep} builds
prefixes of GloVe from $9$ MB (inside one socket's L3) to $530$ MB, and times a
float32 GEMM on the same prefix with the same order of pair count.

\begin{table}[t]
\centering\small
\caption{Build distance throughput (G evaluations/s, 64 threads) versus
working set, GloVe prefixes. GEMM: OpenBLAS $X_{[:n]}X_{[:b]}^{\!\top}$ with
$n\,b\approx 2{\cdot}10^9$ pairs; its rate rises with $n$ only because the
matrix shape improves, reaching $\approx 0.9$ TFLOP/s. Read the ratios at
$n \ge 2{\cdot}10^5$.}
\label{tab:sweep}
\resizebox{\linewidth}{!}{%
\begin{tabular}{@{}rrrrrrr@{}}
\toprule
$n$ & MB & \graft{} T16 & Vamana & GEMM & GEMM/\graft & GEMM/Vamana \\
\midrule
20,000 & 9 & 0.250 & 0.291 & 1.47 & 5.9 & 5.1 \\
50,000 & 22 & 0.226 & 0.338 & 1.48 & 6.5 & 4.4 \\
100,000 & 45 & 0.193 & 0.276 & 2.34 & 12.1 & 8.5 \\
200,000 & 90 & 0.171 & 0.238 & 3.41 & 19.9 & 14.3 \\
500,000 & 224 & 0.167 & 0.220 & 4.54 & 27.2 & 20.6 \\
1,183,514 & 530 & 0.165 & 0.216 & 4.45 & 27.0 & 20.6 \\
\bottomrule
\end{tabular}}
\end{table}

\begin{figure}[t]
\centering
\includegraphics[width=0.92\linewidth]{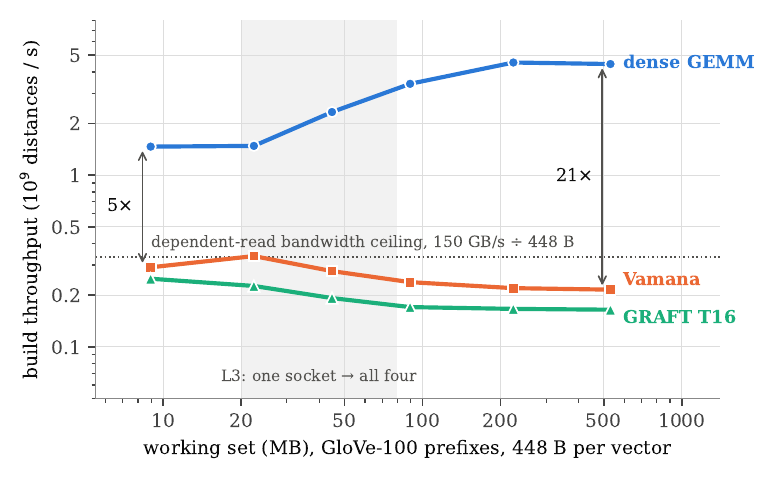}
\caption{The wall, as a function of the working set (Table~\ref{tab:sweep}).
The search builders lose a quarter to a third of their throughput between the
in-cache and the DRAM-resident regime and never approach the dense kernel;
Vamana at full $n$ sits at $65\%$ of the dependent-read bandwidth ceiling.
The GEMM curve rises only because its matrix shape improves with $n$.}
\label{fig:sweep}
\end{figure}

The search rate does not collapse (Figure~\ref{fig:sweep}): from in-cache to
DRAM-resident, \graft{}
loses $34\%$ and Vamana $26\%$, gradually. Already in cache, both evaluate one
distance per $\approx 200$--$250$ core cycles, against $\approx 50$ for a
register-blocked kernel at the efficiency OpenBLAS reaches here. That first
$5$--$6\times$ is \emph{adaptivity overhead}: every evaluation sits inside a
dependency chain (beam heap, visited stamp, best-$R$ tracker, occlusion test)
that must complete before the next vertex can be chosen, so it runs at
latency speed wherever the operand lives. The DRAM boundary adds the remaining
$1.3$--$1.5\times$. At full $n$, Vamana runs at $65\%$ and \graft{} at $49\%$
of the bandwidth ceiling for dependent random reads of one padded vector
($150$ GB/s / $448$ B $= 0.335$ G/s). Two independently engineered search
builders within $1.3\times$ of each other, both $20$--$27\times$ below GEMM at
every $n$, rule out an implementation artifact. These ratios are for
$d = 100$; Section~\ref{sec:dim} measures how they move with the dimension.

\subsection{The model}\label{sec:model}

Write $F$ for the machine's dense arithmetic rate (flop/s achieved by GEMM),
$\mathit{BW}$ for its random-access bandwidth, $d$ for the dimension and
$4d$ for the bytes of a vector. A builder whose pair set $\{(p,q)\}$ is fixed
before evaluation can tile it, so each loaded vector serves a whole tile: its
rate is $F/2d$, and its I/O is $\Theta(\text{pairs}/\sqrt{M})$ for a cache of
$M$ words by the Hong--Kung bound~\cite{hong1981io}. A builder whose next $q$
is a function of the last result cannot tile: each evaluation is one
dependent load, its I/O is $\Theta(\text{pairs})$, and its rate is bounded by
\[
\min\!\Big(\frac{\mathit{BW}}{4d},\ \frac{f_{\text{clk}}}{c_{\text{adapt}}}\Big),
\]
where $c_{\text{adapt}}$ is the per-evaluation bookkeeping in cycles, measured
here at $200$--$250$. The ratio of the ceilings is $2F/\mathit{BW}$ for the
bandwidth term, independent of $d$; with the measured $F$ it is $13\times$ on
this machine, and the two ceilings happen to be of similar height ($0.34$ and
$0.29$ G/s). This is the roofline~\cite{williams2009roofline} with the
arithmetic intensity of adaptive evaluation pinned at its minimum, and the
old memory wall~\cite{wulf1995memorywall} in its random-access form. The gap
widens on newer hardware: wider vector units, quantized GEMM, and GPUs raise
$F$ faster than $\mathit{BW}$ and do nothing for $c_{\text{adapt}}$.

\begin{claim}[The random-access wall]
A graph construction whose distance evaluations are not coherent, because the
distance is a black box, or because the evaluation order depends on previous
results, is bounded by random-access throughput, $\min(\mathit{BW}/4d,\
f_{\text{clk}}/c_{\text{adapt}})$, independently of how few evaluations it
performs. Only a decomposable distance evaluated on a pair set fixed in
advance is bounded by arithmetic.
\end{claim}

\subsection{The wall against dimension}\label{sec:dim}

Every number so far comes from $d = 100$ or $d = 128$, and the model above
invites a conjecture that would limit the result: if the dense side runs at
$F/2d$ and the adaptive side at $\mathit{BW}/4d$, both fall as $1/d$, their
ratio is the $d$-independent $2F/\mathit{BW}$, and the $d$-independent
bookkeeping term $c_{\text{adapt}}$ becomes relatively cheaper as $d$ grows
--- so the wall should \emph{shrink} at embedding dimensions, which is where
it would matter most. We measured it, on an otherwise idle machine, and the
conjecture is false in both of its halves.

\begin{table}[t]
\centering\small
\caption{Cost per pair evaluation against dimension, 64 threads on the idle
machine. Dense is a float32 OpenBLAS product of the same shape at each $d$;
adaptive is a beam's rate, recovered as (queries/s)~$\times$~(evaluations per
query) at recall $\approx 0.95$, so that no new instrumentation enters.
$W$ is the ratio.}
\label{tab:dim}
\resizebox{\linewidth}{!}{%
\begin{tabular}{@{}lrrrrr@{}}
\toprule
 & $d$ & dense cyc/pair & \% of FMA peak & adaptive cyc/pair & $W$ \\
\midrule
Vamana R64 (Deep)   & 96  & 24.6  & 24\% & 368   & 15.0 \\
Vamana R100 (GloVe) & 100 & 22.8  & 27\% & 396   & 17.3 \\
Vamana R64 (SIFT)   & 128 & 25.6  & 31\% & 491   & 19.2 \\
--- & 256 & 38.5 & 42\% & --- & --- \\
--- & 384 & 49.3 & 49\% & --- & --- \\
--- & 768 & 81.9 & 59\% & --- & --- \\
\textbf{Vamana R64 (GIST)} & \textbf{960} & 105.6 & 57\% & \textbf{3,960} & \textbf{37.5} \\
\midrule
\graft{} beam (GloVe) & 100 & 22.8 & 27\% & 782 & 34.3 \\
\textbf{\graft{} beam (GIST)} & \textbf{960} & 105.6 & 57\% & \textbf{8,765} & \textbf{83.0} \\
\bottomrule
\end{tabular}}
\end{table}

Table~\ref{tab:dim} fits two exponents. The adaptive side is \emph{linear} in
the dimension --- $d^{1.04}$ for Vamana's beam at matched degree $R{=}64$
between $d{=}128$ and $d{=}960$, $d^{1.07}$ for ours between $100$ and $960$
--- because it streams $4d$ bytes per dependent access and can block none of
them. The dense side is \emph{sublinear}, $d^{0.63}$, because a blocked
kernel's efficiency climbs with the dimension: from $27\%$ of the machine's
FMA peak at $d{=}100$ to $57\%$ at $d{=}960$. Per additional dimension the
adaptive side pays about $4.1$ cycles and the dense side about $0.10$. Hence
\[
  W(d) \;\propto\; d^{\,1.04-0.63} \;=\; d^{0.4},
\]
and the measured ratios confirm it: $W(960)/W(128) = 1.96$ at matched degree,
$2.16$ and $2.42$ for the unmatched pairings, and $1.51$ for a build-side
pairing of the same two corpora.

The conjecture failed because its first premise did. The bookkeeping term
does decay relative to the arithmetic exactly as the model says; it is simply
not what dominates at high $d$. What dominates is that at $d\approx100$
\emph{both} sides are partly overhead-bound, and therefore closer than they
will ever be again --- high dimension is where a blocked kernel finally gets
to be what it is, and where an adaptive one has nothing left to hide behind.

The consequence for this paper's headline is a qualification that runs the
favourable way: the $20$--$27\times$ measured at $d\approx100$ is a
\emph{conservative} figure for the $384$--$1536$-dimensional corpora that
dominate current practice, where the same measurement gives $37.5$ for a
tuned Vamana and $83$ for our own beam. It should nonetheless always be
quoted with its dimension attached.

\subsection{Can the beam be batched after the fact?}\label{sec:density}

The obvious escape is to run the beams of a block of neighboring points in
lockstep and evaluate the block against the union of their candidates as a
dense product. Its worth is decided by one number, the \emph{density}
$\sum_i |V_i| / (B\,|\bigcup_i V_i|)$, where $V_i$ is the set of vertices
whose distance to point $i$ the beam evaluates: at a $27\times$ kernel
advantage, density must exceed $3.7\%$ to break even. We logged $V_i$ for
$64$ \emph{coherent} blocks (maximal tree-0 SAT subtrees of at most $512$
points, median $178$) and $64$ random blocks of the same size.

\begin{table}[t]
\centering\small
\caption{Lockstep density of the existing harvest. $|V_i|$ = evaluations
per point; union over the block.}
\label{tab:density}
\begin{tabular}{@{}llrrrr@{}}
\toprule
harvest & block & $B$ & $|V_i|$ & $|\bigcup V_i|$ & density \\
\midrule
GloVe T16/\efh 400 & coherent & 178 & 11,547 & 384,776 & \textbf{0.028} \\
GloVe T16/\efh 400 & random & 178 & 12,427 & 875,260 & 0.014 \\
GloVe T16/\efh 100 & coherent & 178 & 3,917 & 191,922 & 0.020 \\
SIFT T4/\efh 400 & coherent & 214 & 2,742 & 86,375 & \textbf{0.032} \\
SIFT T4/\efh 400 & random & 214 & 2,823 & 406,245 & 0.007 \\
\bottomrule
\end{tabular}
\end{table}

Density is $2$--$3\%$ in every configuration (Table~\ref{tab:density}), and
larger blocks make it worse, since the union saturates at $n$. Coherence helps
only $2\times$ on GloVe and $4.5\times$ on SIFT. The per-phase profile removes
the last hope: in expansions $1$--$10$, right after the shared entry set,
adjacent points share $1.2\%$ of their evaluations. The roots have degree in
the hundreds and the argmin over them differs from point to point, so the
beams diverge immediately. What \emph{does} repeat is the target: in coherent
blocks, two thirds of all evaluations go to vertices that at least eight
queries of the block also evaluate. A column-batched kernel could load each
target once per $8$--$30$ uses and remove the bandwidth term, but not
$c_{\text{adapt}}$; its gain is bounded near $2\times$. Finally, the SAT
subtrees are not compact in the scaffold: a coherent block's one-hop union is
$\approx B \times \text{degree}$ (siblings' neighborhoods barely overlap; a
SAT subtree is a cone, not a ball), so the structural candidate sets that
would make a dense product exact by construction are not available from the
forest either.

The conclusion is not that dense construction is impossible for \graft-like
designs, but that it cannot be reached by refining the search. It requires
the pair set to be fixed before evaluation, which means leaves from a
distance-based clustering rather than tree subtrees, plus, for the long-range
edges that dense leaves cannot supply and that put PiPNN's graph
$41$--$45\%$ behind Vamana, a separate and cheap adaptive pass. That is a new
construction and a new baseline, outside this paper.

\section{Negative result: the scaffold's degree}\label{sec:cap}

Before batching the feedback we tried the direct lever on harvest cost,
capping the scaffold's out-degree at $K$ nearest (GloVe, T32/\efh 600,
laptop, distance counts). Cap $32$ cuts build work by $36\%$ but \dq{0.97}
degrades monotonically ($+3\%$ at $K{=}64$, $+24\%$ at $K{=}32$) and
completeness falls from $0.590$ to $0.543$; at matched build cost, lowering
\efh{} beats capping at every point of the sweep ($50.3$\G{} at \efh 400 gives
\dq{0.97} $14{,}047$ where $52.6$\G{} at $K{=}64$ gives $14{,}419$). Routability
needs most of the union's redundancy, which quantifies the substrate control
of~\cite{chavez2026graft}, and locates the gain of \S\ref{sec:results} in the
substrate's \emph{staleness}, not its size.

\section{Related work}\label{sec:related}

\paragraph{Search-based construction} The proximity-graph lineage runs from
the relative neighborhood graph~\cite{toussaint1980rng} and its sparse
navigable relaxations~\cite{arya1993rng} to the navigable small
world~\cite{malkov2014nsw} and HNSW~\cite{malkov2020hnsw}, NSG~\cite{fu2019nsg}
and Vamana/DiskANN~\cite{subramanya2019diskann}; FANNG~\cite{harwood2016fanng}
introduced the occlusion rule that all of them, and \graft, use to sparsify a
candidate set. These builders are \emph{incremental}: each inserted point
searches the graph built so far and mutates it, which is where the feedback
loop of \S\ref{sec:method} comes from and also where their serial dependency
chain comes from. ParlayANN~\cite{manohar2024parlayann} batches the
\emph{insertions} of Vamana and HCNNG into rounds and obtains determinism and
parallel scaling; \fgraft{} batches the \emph{feedback} into a substrate that
begins as a forest union and is frozen within each block, and quantifies what
the feedback buys against what the substrate must supply. Surveys of the
graph family~\cite{wang2021survey} report neither build-side distance counts
nor cost per evaluation, the two quantities this paper measures.

\paragraph{Construction without search} A second lineage fixes its pair sets
before evaluating them. NN-descent~\cite{dong2011nndescent} iterates
neighbor-of-neighbor comparisons over a whole graph in rounds;
EFANNA~\cite{fu2016efanna} and SPTAG~\cite{sptag2018} seed from randomized
$k$-d or balanced trees; HCNNG~\cite{munoz2019hcnng} and
PiPNN~\cite{rubel2026pipnn} cluster hierarchically and brute-force the leaves,
PiPNN with dense rank-update GEMMs; CAGRA~\cite{ootomo2024cagra} builds on a
GPU from a $k$-NN graph, and the GPU $k$-NN graph itself is a blocked
product~\cite{johnson2021gpu}. Every one of these evaluates distances at
arithmetic rate because the pairs are known before they are computed, and
every one of them needs the distance to decompose into inner products (or,
after quantization~\cite{jegou2011pq}, into table lookups). \S\ref{sec:pipnn}
measures the price they pay in graph quality and the price the other lineage
pays in throughput.

\paragraph{Metric-space indexing} The rootstock of \graft{} is the spatial
approximation tree~\cite{navarro2002sat} and its output a half-space proximal
graph~\cite{chavez2005hsp}; both belong to the metric-space setting
where the distance is a black box and the number of evaluations is the cost
model~\cite{chavez2001searching}. The wall of \S\ref{sec:wall} is the
statement that this cost model is exact for the black box and off by an order
of magnitude for vectors.

\paragraph{Performance models} The bandwidth term is the memory
wall~\cite{wulf1995memorywall}; the two-ceiling reading is the roofline
model~\cite{williams2009roofline}; the reuse argument that separates fixed
pair sets from adaptive ones is Hong and Kung's I/O bound for matrix
products~\cite{hong1981io}. Benchmarks measure indexes by query
throughput~\cite{aumuller2020annbenchmarks} and explain their difficulty by
local intrinsic dimensionality~\cite{aumuller2021lid}; to our knowledge no
prior work applies either performance model to index construction.

\section{Conclusion}\label{sec:conclusion}

Incremental construction's advantage over a frozen substrate is a feedback
loop, and the loop is separable from the serial dependency chain that
carries it: batching it into eight synchronous blocks recovers the quality of
four times as many frozen trees, does less distance work than tuned Vamana,
and keeps the graph a pure function of its inputs. The price is a denser
substrate and a higher cost per evaluation, and the result is $1.55\times$
Vamana's build time on GloVe, $1.73\times$ on SIFT, at the quality of the
frozen quality profile.

The larger finding is what bounds that cost. Every search-based builder,
including this one and the incremental ones it was measured against,
evaluates distances as dependent random accesses and runs an order of
magnitude below the machine's dense arithmetic, most of the gap present with
the data in cache, and the gap widening as $d^{0.4}$ --- so the corpora where
it costs most are the $384$--$1536$-dimensional embeddings that dominate
current practice, where we measure it at $37.5\times$ rather than $19\times$.
The beam cannot be batched after the fact. For a black-box
metric this is simply the cost model of the field and the count of
evaluations is the right measure. For vectors it is a wall, and the systems
that stand on the other side of it, PiPNN among them, do so by fixing their
pair sets before they evaluate them, not by evaluating fewer.

\section*{Reproducibility}
Code: \texttt{graft-ann} (public; the mutation, logging and counter patches
are in the paper's companion repository), the ParlayANN and PiPNN patches
(one printed counter; three added counters), the paired driver, the sweep and
density scripts, and every manifest and log behind the tables. Seeds are
fixed; every \graft-family arm produced an identical distance count and
completeness in all ten blocks.

\begingroup\sloppy\hbadness=10000
\bibliographystyle{elsarticle-num}
\bibliography{fgraft}
\endgroup

\end{document}